\documentclass [a4paper]{paper}

\pdfoutput=1

\usepackage {hyperref,cite}
\usepackage [usenames] {color}
\usepackage {plaatjes}
\usepackage {balance}
\usepackage {amssymb,amsmath}

\newtheorem{theorem}{Theorem}
\newtheorem{lemma}[theorem]{Lemma}
\newtheorem{corollary}[theorem]{Corollary}

\newenvironment{proof}{\textbf {Proof:}}{\hfill $\boxtimes$}

\newcommand{\aff}{\mathop{\mathrm{aff}}}

\newcommand{\new}[1]{#1}

\title{Bounds on the Complexity of Halfspace Intersections when the Bounded Faces have Small Dimension}

\author{
David Eppstein\thanks{Computer Science Dept, University of California, Irvine. \url{eppstein@uci.edu}}
\and
Maarten L\"offler\thanks{Computer Science Dept, University of California, Irvine. \url{mloffler@uci.edu}}
}

\begin{document}

\maketitle
\begin{abstract}
We study the combinatorial complexity of $D$-dimensional polyhedra defined as the intersection of $n$ halfspaces, with the property that the highest dimension of any bounded face is much smaller than $D$.
We show that, if $d$ is the maximum dimension of a bounded face, \new{then} the number of vertices of the polyhedron is $O(n^d)$ and the total number of bounded faces of the polyhedron is $O(n^{d^2})$. For inputs in general position the number of bounded faces is $O(n^d)$. 
For any fixed $d$, we show how to compute the set of all vertices, how to determine the maximum dimension of a bounded face of the polyhedron, and how to compute the set of bounded faces in polynomial time, by solving a polynomial number of linear programs.
\end{abstract}

\section{Introduction}

Bounds on the complexity of halfspace intersections and convex hulls~\cite{McM-Mka-70,Sei-CGTA-95}, and algorithms for constructing halfspace intersections and convex hulls~\cite{AviFuk-DCG-92,BarDobHuh-TOMS-96,ChaKap-JACM-70,Cha-DCG-93,Dye-MOR-83,MatRub-MOR-80,Sei-STOC-86,Swa-Algs-85}, have long been a mainstay in discrete and computational geometry. However, because of their inherent exponential dependence on the dimension of the input, these worst case bounds are only useful for polytopes of low to moderate dimension. In contrast, linear programming allows single vertices of halfspace intersections to be found numerically in polynomial time even for high-dimensional inputs~\cite{Kar-Comb-84,Kha-SMD-79,SpiTen-JACM-04}. Is there an intermediate family of polytope construction problems, with low complexity even for high-dimensional input, but capable of \new{describing} more complex sets of features than a single polytope vertex?

\eenplaatje [scale=1.5] {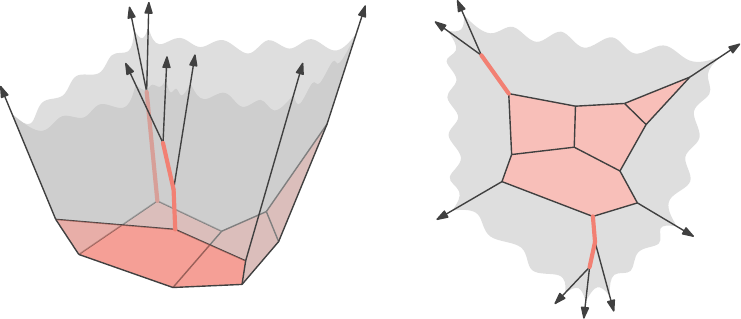} {An unbounded 3-dimensional polyhedron, with a bounded subcomplex consisting of four 2-dimensional and three 1-dimensional maximal faces. On the right we show the projection onto a plane that is perpendicular to the unbounded direction.}

Given the obstacles that have been identified to finding a truly output-sensitive halfspace intersection algorithm~\cite{AviBreSei-CGTA-97}, it may not be sufficient to ask for an output-sensitive algorithm that works efficiently whenever an input polytope has low combinatorial complexity. Instead, in this paper, we identify a different parameter of instance complexity that does lead both to low complexity and efficient algorithms: the dimension of the bounded faces of a halfspace intersection.
\new {Figure~\ref {fig:intro} shows an example polyhedron whose bounded faces have been marked.}
These faces form a combinatorial structure that can be described equivalently in several ways:
\begin{itemize}
\item The \emph{bounded subcomplex} of an intersection of $n$ halfspaces is the cell complex formed by the set of bounded faces of the intersection.
\item The subset of faces of a convex polytope that lie above some linear threshold is equivalent to a bounded subcomplex under a projective transformation that takes the threshold hyperplane to the hyperplane at infinity. \new{A transformation of this type necessarily transforms the faces of the polytope that are entirely above the threshold into bounded faces of the transformed polyhedron, but the faces that cross the threshold (if any exist) become unbounded.}
\item The subset of faces of a convex polytope that are disjoint from some specified facet is again equivalent \new{to a thresholded intersection or a bounded subcomplex}. To see this, one may consider a linear threshold determined by a linear function that is zero on the specified facet, positive on the points of the polytope disjoint from the facet, and negative on the other side of the facet.
\item The complexes considered above are equivalent under projective duality to the subset of faces of a convex polytope that are nonadjacent to some specified vertex.
\end{itemize}
We show that, when the bounded subcomplex consists only of faces of dimension at most $d$, where $d$ is significantly smaller than the ambient dimension, then it has smaller combinatorial complexity than the bounds given by the upper bound theorem. Specifically, it has $O(n^d)$ vertices, and $O(n^{d^2})$ faces. \new{More strongly,} when the halfspaces forming the intersection are in general position, we obtain a tighter bound of $O(n^d)$ on the number of faces as well. Based on these combinatorial bounds, we provide algorithms for listing all the vertices or faces of the bounded subcomplex in polynomial time whenever $d$ is a fixed constant.

\new{\section{Motivating Applications and Related Work}\label{sec:motive}}
One motivation for studying bounded subcomplexes comes from the \emph{tight span} construction, a canonical method of embedding any metric space into a continuous space with properties similar to those of $L^\infty$ spaces~\cite{ChrLar-JoA-94,Dre-AiM-84,Isb-CMH-64}. One way of defining the tight span, for a \new{finite} metric space with $n$ points $p_i$, is \new{to coordinatize $n$-dimensional $L^\infty$ space by $n$ variables $x_i$, and to define a polyhedral subset of the space by the \new{$\binom{n}{2}$} linear inequalities
\[
x_i+x_j\ge \mathop{\mathrm{dist}}(p_i,p_j)
\]
for each possible pair $(i,j)$. Then, the tight span is the bounded subcomplex of this polyhedron.}
For metric spaces satisfying an appropriate general position assumption, the dimension of the bounded subcomplex is between $\lceil n/3\rceil$ and $\lfloor n/2\rfloor$~\cite{Dev-AoC-06}, but certain combinatorially defined metrics, such as the metrics of distances on certain classes of planar graphs, can have tight spans of much lower dimension~\cite{Epp-TAA-09}. Our results bound the complexity of these low-dimensional tight spans and allow them to be constructed efficiently, generalizing our previous algorithms for constructing tight spans \new{when they are} homeomorphic to subsets of the plane~\cite{Epp-09}.

Our results can also be interpreted as statements about the complexity of Delaunay triangulations for inputs satisfying strong convex position assumptions. Delaunay triangulations are closely related to convex hulls: the Delaunay triangulation is combinatorially equivalent to the convex hull of a point set lifted to a sphere in one higher dimension, augmented by an extra point at the pole of the sphere, followed by the removal of all faces incident to that pole~\cite{Bro-IPL-79}.  If \new{a $D$-dimensional set of $n$ points has the property that every interior point of the convex hull of  the set belongs to a Delaunay triangulation feature of dimension at least $D-d$, then our results imply via this lifting relation that the Delaunay triangulation has $O(n^d)$ $D$-dimensional simplices. For instance, if the convex hull of a point set is a stacked polytope (a convex figure formed by gluing simplices facet-to-facet) and the Delaunay triangulation coincides with the gluing pattern of the simplices, then every interior point of the convex hull belongs either to one of the simplices or to one of the glued facets, so $d=1$; in this case, the number of simplices in the Delaunay triangulation is exactly $n-D$.}

Additionally, many combinatorial optimization problems such as shortest path trees, minimum spanning trees, bipartite minimum weight perfect matchings, and network flows can be expressed as linear programs, and the vertices of the corresponding polyhedra that have the $k$ smallest values (according to the linear objective function used for the problem) correspond to the best $k$ solutions of these optimization problems. From this point of view, it is of interest to understand the combinatorial behavior of the subsets of polyhedron vertices and faces satisfying a linear threshold constraint, or equivalently of the bounded faces of polyhedra. As a somewhat trivial example (easy to prove directly without resorting to our polyhedral point of view), consider the bin packing problem in which the input consists of a set of items with weights $w_i$ and a capacity $W$, and a solution consists of a subset of the items with total weight at most $W$. If $x_i$ is an indicator variable for the inclusion of item $i$, then the solutions are the vertices of the hypercube $0\le x_i\le 1$ that satisfy the linear constraint
\new{\[\sum w_i x_i\le W.\]
The number of halfspaces defining the hypercube is $2n$, and the dimension $d$ of the bounded subcomplex of the hypercube determined by the linear constraint is the maximum number of items that can be packed into a single solution. Therefore, our results imply that the total number of solutions is $O(n^d)$.}

Bounded subcomplexes have also been investigated in other contexts. 
In tropical geometry, the \emph{tropical complex} defined by a set of points is also equivalent to the set of bounded faces of a polytope defined from the points~\cite{DevStu-DM-04}. Develin~\cite{Dev-DCG-04} studied the bounded faces of a certain polytope arising from a problem in algebra, and showed that they are all isomorphic to subpolytopes of permutohedra. Queyranne~\cite{Que-MP-93} used polytopes to model scheduling problems; the polytope defined by Queyranne has a unique bounded facet. In connection with the tight span application, Hirai~\cite{Hir-AC-06} showed that the bounded faces of a halfspace intersection form a contractable complex. Herrmann et al.~\cite{HerJosPfe-10} consider the problem of computing the bounded subcomplex of a polyhedron, but they do not bound the complexity of the complex. Their algorithms assume that the vertices and facets of the polyhedron are both already known, and are output-sensitive given this information.

\section{Definitions} \label {sec:defs}

Define a \emph{polytope} to be the convex hull $CH(S)$ of a finite set $S$ of points, and a \emph{polyhedron} to be the intersection $\cap T$ of a finite set $T$ of closed halfspaces. Every polytope is a polyhedron; a polyhedron is a polytope if and only if it is bounded.

The \emph{faces} of a polytope or polyhedron $P$ are the sets of the form $P\cap H$ where $H$ is a closed halfspace whose boundary is disjoint from the relative interior of $P$. Every face of a polytope is also a polytope, and every face of a polyhedron is also a polyhedron. The \emph{dimension} $\dim f$ of a face~$f$ is the dimension of its affine hull $\aff f$. With these definitions, $P$ itself is a face, as is the empty set. By convention the dimension of the empty set is~$-1$. Two faces are \emph{incident} if one is a subset of the other.
A face $f$ of a polytope $CH(S)$ may be identified with the set $f\cap S$, and a face $f$ of a polyhedron $\cap T$ may be identified with the set of halfspaces of $T$ whose boundary hyperplanes contain~$f$. This identification is one-to-one: each face is identified with a unique set. The subset relation between faces is either the same as or the reverse of the subset relation between the sets the faces are identified with, for polytopes and polyhedra respectively.

A \emph{vertex} of a polytope or polyhedron is a 0-dimensional face; it must be a singleton set, containing a single point. An \emph{edge} of a polytope or polyhedron is a bounded 1-dimensional face, consisting of a line segment connecting two vertices. Polyhedra may also have \emph{rays} and \emph{lines}, unbounded 1-dim\-en\-sion\-al faces with one or zero vertices on them respectively. If a polytope or polyhedron $P$ has dimension $D$, then a \emph{facet} of $P$ is a $(D-1)$-dimensional face. If $D$ is also the dimension of the space containing the points or halfspaces from which $P$ was defined, then $P$ is said to be \emph{full-dimensional}.
More generally we define \new{a} \emph{$d$-face} of $P$ to be a face of dimension $d$, so for instance a $1$-face may be an edge, a ray, or a line.

\eenplaatje [scale=1.5]{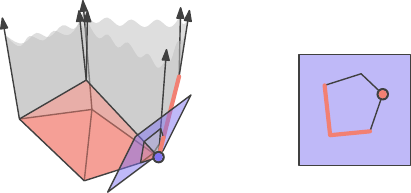} {The link of a vertex of $P$ is a polytope of dimension one lower than $P$.}

If $v$ is a vertex of a polytope or polyhedron $P$, let $h$ be a hyperplane that separates $v$ from all other vertices and all unbounded regions of $P$, and define the \emph{link} of $v$ to be the intersection $P\cap h$. \new {Figure~\ref {fig:link} shows an example.} Different choices of $h$ will give geometrically different intersections, but they are all related to each other by affine transformations, so as a combinatorial polytope the link is well defined; however, in our proofs we will need to refer to the link as a geometric object, defined as above, rather than as a combinatorial object, so in each such usage we will need to verify that the choice of $h$ does not affect the result. The faces of the link correspond one-for-one, in an incidence-preserving way, with the faces of $P$ that are incident to $v$.

A full-dimensional polytope or polyhedron is \emph{simple} if every vertex is incident to at most $D$ edges. Equivalently, a polytope or polyhedron is simple if the link of every vertex is a simplex. A polyhedron defined as the intersection of a set of halfspaces in general position (meaning that a small perturbation of any of the halfspaces does not change the combinatorial structure of the polyhedron) is necessarily a simple polyhedron. Dually, a polytope is \emph{simplicial} if all of its faces except for the polytope itself are simplices. A polytope defined as the convex hull of a set of points in general position is necessarily simplicial.

If $P$ is a polytope containing the origin, the \emph{polar polytope} of $P$ is the intersection of a system of halfspaces, one for each vertex of $P$: if $v$ is a vertex of $P$ then the corresponding halfplane is the set of points $w$ such that $v\cdot w\le 1$. Although defined as an intersection of halfspaces, the polar polytope is bounded and is therefore a polytope. The faces of the polar polytope correspond one-for-one with the faces of $P$ but in a dimension-preserving way: an $i$-dimensional face of $P$ corresponds to a $(D-i-1)$-dimensional face of the polar polytope. The polar polytope of the polar polytope is $P$ again. If $P$ is simple, its polar polytope is simplicial, and if $P$ is simplicial, its polar polytope is simple.

If $P$ is a polyhedron or polytope and $\ell$ is a linear function, we define $\max_\ell P$ and $\min_\ell P$ to be the face on which $\ell$ takes its minimum or maximum. If the minimum or maximum is unbounded, then we define $\max_\ell P$ or $\min_\ell P$ to be the empty set.
A polyhedron is \emph{pointed} if it has at least one vertex. If $\ell^0$ is a tangent hyperplane to a polyhedron $P$ at a given vertex $v$, that is not tangent to any higher-dimensional face, and $\ell$ is a linear function that is positive on the side of $\ell^0$ containing $P$ and negative on the other side, then $v=\min_\ell P$, so a polyhedron $P$ is pointed if and only if there exists a linear function $\ell$ for which $\min_\ell P$ is a single vertex.

A \emph{polyhedral complex} is a finite set $C$ of polyhedra, all in the same ambient space, such that $C$ contains each face of each polyhedron in $C$ and such that the intersection of any two polyhedra in $C$ is a face of both. We define a \emph{polytopal complex} to be a polyhedral complex in which each member of $C$ is a polytope.

If $P$ is a polytope or polyhedron, $\ell$ is a linear function, and $B$ is any real number, let $P^{\ell<B}$ be the polyhedral complex formed by the set of faces of $P$ such that, for every point $p$ of a face in $P^{\ell<B}$, $\ell(p)<B$.
Analogously, for $B=\infty$, define $P^{\ell<\infty}$ to be the set of faces of $P$ on which $\ell$ is bounded. (This notation directly implies only that $\ell$ is bounded from above, but we define it to mean that it is bounded from below as well.)
We define $\dim P^{\ell<B}$ to be the maximum dimension of a face in $P^{\ell<B}$.

\section{Examples}

\begin{figure}[t]
\centering\includegraphics[width=3.75in]{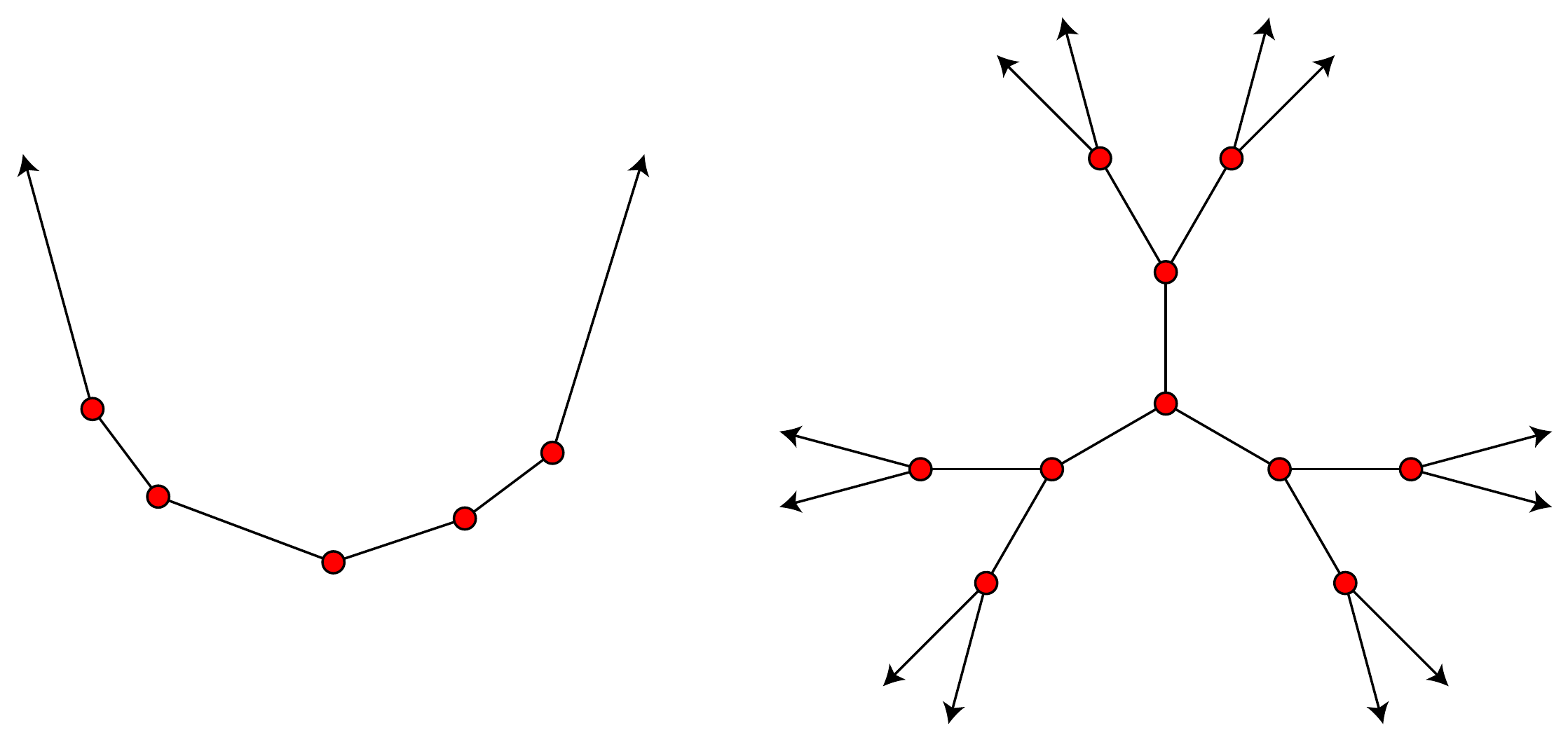}
\caption{Trees that may be realized as the bounded faces of two- and three-dimensional polyedra.}
\label{fig:tight-trees}
\end{figure}

\new{We begin with some examples of polyhedra with low-dimensional bounded faces. These are not so much intended as motivation for the bounded-face problem but rather as illustrations of the sort of behavior these polyhedra may have; for more in-depth motivating examples, see Section~\ref{sec:motive}.}

\begin{itemize}

\item Let $P$ be any $(D-1)$-dimensional polytope,
having facets determined by the linear inequalities $\bar a_i \cdot x\ge b_i$. Then the $D$-dimensional polyhedron defined by the linear inequalities $\bar a_i\cdot x - b_i y\ge 0$ is an unbounded cone over $P$: it has $P$ as its cross-section (say in the hyperplane $y=1$) but has only a single vertex at the origin. Examples of this form can be used to show that, even when $\dim P^{\ell<B}$ is small, the number of unbounded faces can be large. For instance, if $P$ is a hypercube, defined as the intersection of $n=2(D-1)$ facets, then the cone over $P$ has $3^{n/2}+2$ faces despite having only a single bounded vertex.

\item For $d=1$ and $D=2$, an unbounded two-dimensional polyhedron with $n$ facets always has exactly $n-1$ vertices (Figure~\ref{fig:tight-trees}). For $d=1$ and $D=3$, if $T$ is a tree forming the boundary edges and rays of a partition of the plane into $n$ unbounded polygons~\cite{CarEpp-GD-06}, with exactly three polygons meeting at each vertex, then $T$ may be lifted to an unbounded three-dimensional polyhedron with $n$ facets and $n-2$ vertices (Figure~\ref{fig:tight-trees}).

\item More generally, a $D$-dimensional polyhedron may have any free tree $T$ with vertex degree at most $D$ as its complex of bounded faces. To see this, let $P$ be a $(D-1)$-dimensional \emph{stacked polytope}, the union of a set of simplices glued facet-to-facet with one simplex for each vertex of $T$ and with two simplices glued together whenever the corresponding vertices of $T$ are adjacent. It can be shown by induction on the size of $T$ that there is always a way of choosing simplices to glue in this way such that their union is convex. $P$ can be lifted to a stacked polytope $P^+$ in $D$ dimensions, by adding the same new vertex $v$ to every simplex of $P$; again, $P^+$ can be constructed by gluing simplices in such a way that their union is convex. Let $Q^+$ be the polar polytope of $P^+$, let $Q$ be the polyhedron formed by intersecting all but one of the facet halfspaces of $Q^+$ (omitting the facet dual to $v$), and let $\ell$ be a linear function that is zero on the facet of $Q^+$ dual to $v$ and negative on $Q$. Then $Q^{\ell<0}$ is isomorphic to $T$. Note that, in defining $P$, all points may be chosen to be in general position, from which it follows that in $Q$, all halfspaces may be chosen to be in general position.

\begin{figure}[t]
\centering\includegraphics[width=2.5in]{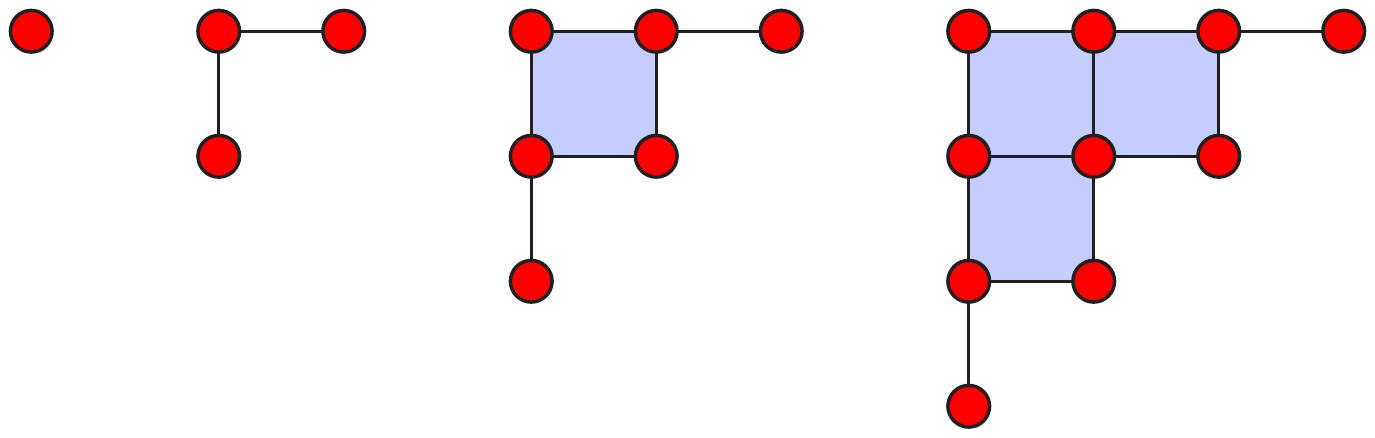}
\caption{The bounded subcomplex of the Voronoi diagram of $n$ points on the three-dimensional moment curve, for $4\le n\le 7$.}
\label{fig:moment-vor}
\end{figure}

\item Any $n$ points on the three-dimensional moment curve $(t,t^2,t^3)$ determine a Voronoi diagram with $\binom{n-2}{2}$ vertices~\cite{DewVra-UM-77}. A standard lifting transformation takes this Voronoi diagram to an unbounded four-dimensional polyhedron, with $n$ defining halfspaces. In this polyhedron, all facets are unbounded, so the maximum dimension of a bounded face is two. There are quadratically many bounded faces (Figure~\ref{fig:moment-vor}). As this example shows, a nontrivial bound on $d$ can still lead to examples with superlinear complexity.

\item Let $P$ be a four-dimensional pyramid over an octahedron, modified by folding its octahedral facet so that it is split into two square pyramids. Then the subset of faces of $P$ that are nonadjacent to the apex of $P$ are dual to a one-dimensional complex that takes the form of a single line segment: the two square pyramid facets are dual to the endpoints of the line segment, and the square ridge separating these two facets is dual to the line segment itself. However, if the vertices of $P$ are perturbed into general position, then the square becomes a flat tetrahedral facet while each of the two square pyramids is split into two tetrahedra; the dual complex becomes two triangles joined at a vertex. This example shows that some instances cannot be perturbed into general position without increasing the dimension of the bounded subcomplex.

\item Let $Q_D$ be the $D$-dimensional hypercube, formed by intersecting $n=2D$ halfspaces determined by the inequalities $0\le x_i$ and $x_i\le 1$, and let $\sigma$ be the linear function $\sum x_i$. Then, for any integer $0\le d\le D$, $\dim Q_D^{\sigma>D-d-1/2}=d$. The number of vertices of $Q_D^{\sigma>D-d-1/2}$ is $\sum_{i=0}^d \binom{D}{i}$. For any even integer $n$, and any constant value of~$d$, this provides an example of a bounded subcomplex in general position with $\Theta(n^d)$ vertices.

\end{itemize}

\section{Bounding the number of vertices}

In order to help prove our main result, we begin with a technical lemma stating that if we slice a polytope by a hyperplane in such a way that one side of the slice contains only faces of low dimension, then the other side of the slice contains a face of high dimension. \new {Figure~\ref {fig:tetraslices} illustrates the lemma.}

\eenplaatje [scale=1.5] {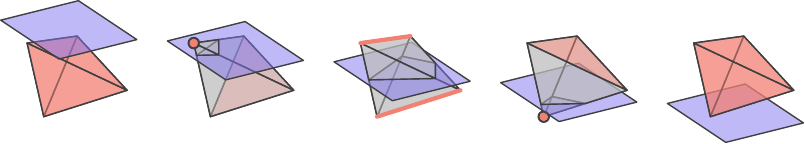} {No matter how we slice a tetrahedron with a plane, there is always a pair of faces (one of which could be the empty set) on both sides of the plane such that the affine hull of these faces is the full space. Lemma~\ref {lem:slice-basis} claims that this is still true for any polytope, in any dimension.}

\begin{lemma}
\label{lem:slice-basis}
Let $P$ be a polytope, and $\lambda$ be a hyperplane that does not pass through any vertex of $P$. Then there exist faces $f^+$ and $f^-$ of $P$, on opposite sides of $\lambda$, such that $\aff(f^+\cup f^-)=\aff P$.
\end{lemma}

\begin{proof}
We use induction on dimension.
\new {Note that the empty set is also a face of $P$, so $\lambda$ need not intersect $P$.}
Let $\ell$ be a linear function that is zero on $\lambda$, let $P^+=P^{\ell>0}$ be the complex of faces of $P$ on which $\ell$ is entirely positive, and let $P^-=P^{\ell<0}$ be the complex of faces of $P$ on which $\ell$ is entirely negative. 
As a base case, if $P^+$ is empty, then we may take $f^+=\emptyset$ and $f^-=P$. 
As a second base case, if $P^+$ is a single vertex $v$, then we may take $f^+=v$ and $f^-$ to be any facet of $P$ disjoint from $v$; such a facet must exist, for otherwise $P$ would be unbounded.

We may assume without loss of generality (by perturbing $\lambda$, if necessary) that no hyperplane parallel to $\lambda$ passes through two or more vertices of $P$. 
By the assumption of general position on $\lambda$, $\ell$ does not take the same value on any two vertices of $P$. Let $v$ be the vertex in $P^+$ minimizing $\ell$, and let $L$ be the link of $v$ (with an arbitrary choice of the intersecting hyperplane $\mu$ defining the link). Let $L^+=L^{\ell>\ell(v)}$ and $L^-=L^{\ell<\ell(v)}$. By the induction hypothesis,
\new {applied to the polytope $L$,}
we may find faces $g^+$ and $g^-$ in $L^+$ and $L^-$ respectively, such that $\aff(g^+\cup g^-)=\aff L$. (When $P$ is full-dimensional, this affine hull is just the hyperplane $\mu$.) The face $g^+$ of $L^+$ corresponds to a face $f^+$ of $P$ in which all points have a value of $\ell$ larger than or equal to $\ell(v)$; since $\ell(v)$ is strictly positive, $f^+$ belongs to $P^+$. However, the face $g^-$ of $L^-$ corresponds to a face $h$ of $P$ such that $v$ lies in $P^+$ and the remaining vertices of $h$ lie in $P^-$. Let $f^-$ be any facet of $h$ disjoint from $v$, as in the second base case.

Then \new{\[g^+=(f^+\cap\mu)\subset f^+\subset\aff(f^+\cup f^-).\]
Additionally, $v\in f^+$, so \[g^-=(h\cap\mu)\subset h\subset\aff(v\cup f^-)\subset\aff(f^+\cup f^-),\] and
\[\aff P=\aff(L\cup v)=\aff(g^+\cup g^-\cup v)\subset\aff(f^+\cup f^-).\]} But $f^+\cup f^-$ is a subset of $P$, so its affine hull cannot be a proper superset of the affine hull of $P$, and the two affine hulls are equal, as desired.
\end{proof}

As we now show, the vertices of a polyhedron with low-dimensional bounded faces can be associated to small sets of facets of the polyhedron, leading to polynomial bounds on the number of vertices of the polyhedron.
\new {Figure~\ref {fig:dmin} illustrates the lemma for $3$-dimensional polyhedra.}

\eenplaatje [scale=1.5] {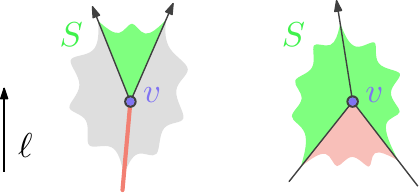} {A vertex $v$ of a 3-dimensional polyhedron has at least three incident edges. On the left, two edges point "up" (in the direction of $\ell$), so $v$ is the minimum of the face between them: $|S| = 1$. In this case, the only face that is necessarily bounded is the edge that points down (which has dimension $1$). On the right, only one edge points up, so $v$ is the minimum of the intersection of the two faces on both sides of it: $|S| = 2$. In this case, however, at least two edges point down, so the face between them is necessarily bounded (which has dimension $2$).}

\begin{lemma}
\label{lem:vertex-set-corr}
Let $P$ be a full-dimensional pointed polyhedron, defined as the intersection of $n$ $D$-dimensional halfspaces. Let $\ell$ be a linear function with a vertex of $P$ as its minimum, and suppose that $\ell$ does not take the same value at any two different vertices of $P$. Let $B$ be \new {a} given \new {real number}, and let $d=\dim P^{\ell<B}$.  Let $v$ be any vertex of $P^{\ell<B}$. Then there exists a set $S$ of $d$ of the halfspaces defining $P$, such that $v=\min_\ell(P\cap A)$, where $A$ is the intersection of the boundary hyperplanes of $S$.
\end{lemma}

\begin{proof}
Define a linear function $\ell_v$ by $\ell_v(x)=\ell(x)-\ell(v)$. Then $\ell_v$ differs from $\ell$ by a simple translation; \new{therefore,} it has the same maxima and minima \new{as~$\ell$} on each face of $P$. Let \new {$V$ be the link of $v$, let} $V^+$ be the complex of faces of \new {$V$} on which $\ell_v$ is positive, and let $V^-$ be the complex of faces of \new {$V$} on which $\ell_v$ is negative; this definition of $V^+$ and $V^-$ does not depend on the intersecting hyperplane used to define the link geometrically.

By Lemma~\ref{lem:slice-basis}, we can find faces $f^+$ and $f^-$ in $V^+$ and $V^-$ respectively, such that 
\new {$\aff(f^+\cup f^-)=\aff V$. Since $\dim V = \dim P - 1 = D - 1$, we have}
$\dim f^+ +\dim f^- \new {\geq} D-2$.
Then these two faces \new{$f^+$ and $f^-$} correspond to faces $g^+$ and $g^-$ in $P$, such that
$v=\min_{\ell} g^+=\max_{\ell} g^-$, and such that $\dim g^+ +\dim g^- \new {\geq} D$.
But we know that $\dim g^-\le d$ by assumption, so $\dim g^+\ge D-d$. Therefore, $g^+$ may be represented as the intersection of $P$ with a set of the boundary hyperplanes of at most $d$ of the defining halfspaces of $P$. If $A$ is the intersection of these hyperplanes, then $v=\min_\ell g^+=\min_\ell(P\cap A)$ as desired.
\end{proof}

\begin{theorem}
\label{thm:num-vertices}
Let $P$ be a full-dimensional pointed polyhedron, defined as the intersection of $n$ $D$-dimensional halfspaces, let $\ell$ be a linear function which attains a minimum on $P$, let $B$ be given and greater than the minimum value of $\ell$, and suppose that $\dim P^{\ell<B}=d<D$. Then the number of vertices of $P^{\ell<B}$ is at most
$$\binom{n}{d}-\binom{D}{d}+1.$$
\end{theorem}

\begin{proof}
Lemma~\ref{lem:vertex-set-corr} gives a formula $v=\min_\ell(P\cap A)$ by which each vertex $v$ of $P$ may be obtained from a subset of at most $d$ of the defining halfspaces of $P$. We may augment any set of halfspaces determining $v$ in this way so that it has exactly $d$ halfspaces, by making an arbitrary choice of a sufficient number of the other halfspaces whose boundaries are incident to $v$. There are $\binom{n}{d}$ sets of exactly $d$ halfspaces, each one determines at most one vertex,
and all vertices can be determined in this way, but some vertices may be duplicated. In particular the vertex $\min_\ell P$ belongs to at least $D$ facets, and is the vertex that is determined by each of the $\binom{D}{d}$ different $d$-tuples of these $D$ halfspaces, giving rise to the correction terms in the formula.
\end{proof}

\begin{corollary}
\label{cor:asymptotic-vertices}
Let $P$ be a polyhedron, defined as the intersection of $n$ $D$-dimensional halfspaces, let $\ell$ be a linear function which attains a minimum on $P$, let $B$ be given, and suppose that  $\dim P^{\ell<B}=d<D$. Then the number of vertices of $P^{\ell<B}$ is $O(n^d)$.
\end{corollary}

\begin{proof}
We may assume without loss of generality that $P$ is full-dimensional, for otherwise we may restrict our attention to the affine hull of $P$ and eliminate any halfspaces that contain the entire affine hull.

If $P$ is not pointed or $B$ does not exceed the minimum value of $\ell$, then $P^{\ell<B}$ has zero vertices and the result is obvious. Otherwise, we may ignore the $-\binom{D}{d}+1$ correction terms in Theorem~\ref{thm:num-vertices}, obtaining a formula which grows asymptotically as $O(n^d)$ and is independent of~$D$.
\end{proof}

Because of the upper bound theorem, these bounds can only be tight when $d\le D/2$. However, as our main interest is for bounded $d$ and unbounded $D$, this is not a significant limitation.
The formula of Theorem~\ref{thm:num-vertices} is tight when $d=0$, and the examples of polyhedra with arbitrary $D$-ary trees as their bounded subcomplex show that it is also tight for $d=1$. The example of the Voronoi diagram of points on the moment curve shows that the asymptotic bound of Corollary~\ref{cor:asymptotic-vertices} is also tight when $d=2$ and $D=4$, and more strongly that the formula of Theorem~\ref{thm:num-vertices} is tight to within a $1-O(1/n)$ factor in this case. The hypercube example $Q_D^{\sigma>D-d-1/2}$ shows that the asymptotic bound of Corollary~\ref{cor:asymptotic-vertices} is tight for any constant~$d$.

\section{Euler's formula}

Letting $\Phi(P^{\ell<B})$ denote the set of faces of $P^{\ell<B}$,
define the \emph{Euler characteristic}
$$\chi(P^{\ell<B})=\sum_{f\in \Phi(P^{\ell<B})}(-1)^{\dim f}.$$
Recall that we include the empty set as a face, as well as $P$ itself in the case that $\max_\ell P<B$. As is well known, with this definition, $\chi(P)=0$ for any convex polytope $P$; Ziegler~\cite{Zie-GTiM-95} gives a nice proof of this fact based on shelling.

As we now show, the Euler characteristic of $P^{\ell<B}$ behaves similarly to the Euler characteristic of an entire polytope. This fact will lead to improved bounds on the number of faces of  $P^{\ell<B}$, because it will imply that the $d$-dimensional faces can be charged against other faces of lower dimension. Our proof can be viewed as a dual form of the shelling proof of Euler's formula for polytopes.

\begin{theorem}
\label{thm:euler}
Let $P$ be a polytope or polyhedron, let $\ell$ be a linear function, and let $B$ be either $\infty$ or a real number. Additionally, suppose that $P^{\ell<B}$ is nonempty. Then $\chi(P^{\ell<B})=0$.
\end{theorem}

\begin{proof}
We use induction, both on the dimension of $P^{\ell<B}$ and on the number of vertices of $P^{\ell<B}$. As a base case, the result is clearly true if $P^{\ell<B}$ consists only of a single vertex: for, in that case, $F$ contains two faces, the empty set of dimension $-1$ and that vertex, and these two faces have opposite signs in the sum defining $\chi(P^{\ell<B})$, cancelling each other out.

Otherwise, $P^{\ell<B}$ has at least two vertices. We may assume without loss of generality (by perturbing $\ell$ if necessary) that no two vertices of $P^{\ell<B}$ have the same value of $\ell$. Let $v$ be the vertex of $P^{\ell<B}$ with the maximum value of $\ell$, and let $B'=\ell(v)$. Then, by induction, $\chi(P^{\ell<B'})=0$.

$P^{\ell<B}$ differs from $P^{\ell<B'}$ by the set of faces incident  to $v$. Let $L$ be the link of $v$ in $P$; then the face structure of $L^{\ell<B'}$ does not depend on the intersecting hyperplane used to define the geometry of $L$, and the faces incident to $v$ in $P^{\ell<B}$ are in one-to-one correspondence with the faces of $L^{\ell<B'}$.
For instance, the empty face of $L^{\ell<B'}$ corresponds to the vertex $v$ of $P^{\ell<B}$, and each vertex $u$ of $L^{\ell<B'}$ corresponds to an edge $uv$ of $P^{\ell<B}$. In this correspondence, the dimension of a face in $L^{\ell<B'}$ is one less than the dimension of the corresponding face of $P^{\ell<B}$. Therefore, each face of $L^{\ell<B'}$ makes a contribution to the Euler characteristic of $L^{\ell<B'}$ with the opposite sign to the contribution of the corresponding face in $P^{\ell<B}$. Therefore,  $\chi(P^{\ell<B})=\chi(P^{\ell<B'})-\chi(L^{\ell<B'})$. But $L^{\ell<B'}$ has lower dimension than $P^{\ell<B}$, and is non-empty (it includes at least one vertex corresponding to an edge that can be reached from $v$ by a single step of the simplex method), so by induction on dimension, $\chi(L^{\ell<B'})=0$. Therefore, $\chi(P^{\ell<B})=\chi(P^{\ell<B'})-\chi(L^{\ell<B'})=0-0=0$.
\end{proof}

As a simpler proof for the special case in which $P$ is a polytope and $f$ is a facet of $P$ on which $\ell=B$, let $\epsilon$ be sufficiently small that there are no vertices of $P$ for which $\ell$ is between $B$ and $B-\epsilon$, and let $\lambda$ be the hyperplane $\ell=B-\epsilon$. Then
$$\chi(P^{\ell<B})=\chi(P) + \chi(P\cap\lambda) - \chi(f)=0+0-0=0.$$

\section{Bounding the number of faces}

\begin{theorem}
Let $P$ be a polyhedron, defined as the intersection of $n$ $D$-dimensional halfspaces, let $\ell$ be a linear function which attains a minimum on $P$, let $B$ be given, and suppose that  $\dim P^{\ell<B}=d<D$. Then the number of faces of $P^{\ell<B}$ is $O(n^{d^2})$.
\end{theorem}

\begin{proof}
By Theorem~\ref{thm:euler}, the number of $d$-dimensional faces is no larger than the total number of faces of all lower dimensions. Therefore, we need only count faces of dimension up to $d-1$. But each such face is the intersection of $P$ with the affine hull of a set of at most $d$ vertices of $P^{\ell<B}$. There are $O(n^d)$ vertices, so there are $O(n^{d^2})$ sets of at most $d$ vertices, and therefore $O(n^{d^2})$ faces.
\end{proof}

For instances in general position, a considerably sharper bound may be obtained.

\begin{theorem}
\label{thm:general-pos-faces}
Let $P$ be a polyhedron, defined as the intersection of $n$ $D$-dimensional halfspaces, let $\ell$ be a linear function which attains a minimum on $P$, let $B$ be given, and suppose that  $\dim P^{\ell<B}=d<D$. Additionally, suppose that the halfspaces defining $P$ are in general position and that $P$ has $N$ vertices. Then the number of $i$-dimensional faces of $P^{\ell<B}$ is at most $N\binom{d}{i}$.
\end{theorem}

\begin{proof}
We charge each $i$-dimensional face $f$ of $P^{\ell<B}$ to the vertex $\max_\ell f$.
Let $v$ be any vertex, and let $L$ be its link. Then the faces charged to $v$ correspond one-to-one with the faces of $L^{\ell<\ell(v)}$.
By the general position assumption, $L$ is a simplex, and therefore $L^{\ell<\ell(v)}$ is also a simplex. Because the faces of $P^{\ell<B}$ correspond to faces with dimension one less in $L^{\ell<\ell(v)}$, $\dim L^{\ell<\ell(v)}\le d-1$. Thus, the number of $i$-dimensional faces  of $P^{\ell<B}$ that are charged to $v$ is at most the number of $(i-1)$-dimensional faces of a $(d-1)$-dimensional simplex, which is $\binom{d}{i}$.
\end{proof}

\begin{corollary}
Let $P$ be a polyhedron, defined as the intersection of $n$ $D$-dimensional halfspaces, let $\ell$ be a linear function which attains a minimum on $P$, let $B$ be given, and suppose that  $\dim P^{\ell<B}=d<D$. Additionally, suppose that the halfspaces defining $P$ are in general position. Then the number of faces of $P^{\ell<B}$ is $O(n^d)$.
\end{corollary}

\begin{proof}
By Theorems~\ref{thm:num-vertices} and~\ref{thm:general-pos-faces}, adding the bounds on the the number of faces of $P^{\ell<B}$ over all dimensions up to $d$ gives a total that is less than $\binom{n}{d}2^d<n^d\frac{2^d}{d!}=O(n^d)$.
\end{proof}

\section{Algorithms}

We now discuss algorithms for finding the bounded subcomplex of a polyhedron. There are many alternative solutions available, depending on whether the dimension $d$ of the subcomplex is known or unknown to the algorithm, on whether only the vertices of the polyhedron need to be found or whether the whole bounded subcomplex is to be constructed, and on whether we use as a subroutine the algorithm for computing bounded subcomplexes of Herrmann et al.~\cite{HerJosPfe-10} or whether we directly enumerate the faces of the bounded subcomplex, using linear programming to test whether each face is bounded.

Our algorithms will necessarily involve the solution of linear programs and linear feasability problems. Especially in the case of inputs that are not in general position, it is important that these linear programs are solved in an exact model of computation that allows solution vertices to be compared for equality: it is possible that a small perturbation of an input problem, as might occur due to round-off error in a non-exact numerical linear programming algorithm, could significantly increase the dimension of the bounded subcomplex. However, this sort of exact computation model is standard in computational geometry algorithms, and solutions are available when the ambient dimension $D$ is of moderate size~\cite{MatShaWel-Algo-96}. Even when $D$ is large, strongly polynomial algorithms are known for some special cases of linear programming, such as the case with two variables per inequality that arises in the tight span construction~\cite{Meg-SJC-83}. We let $L$ denote the time to solve a linear program of ambient dimension $D$ and $n$ constraints, or to determine whether such a program is infeasible or unbounded.

When the input consists of a set of $n$ halfspaces together with the dimension $d$ of the bounded subcomplex, we have the following results:
\begin{itemize}
\item We can construct all vertices of the given polyhedron in time $O(n^d L)$. This follows immediately from Lemma~\ref{lem:vertex-set-corr} which gives a formula $v=\min_\ell(P\cap A)$ allowing each vertex $v$ of $P$ to be obtained as the solution to a linear program in a subspace of dimension $D-d$.
\item If the input is in general position, we can construct the bounded subcomplex in time $O(n^d L + n^{5d})$. This method generates the vertices as above, exhaustively tests each vertex-facet pair to find all the vertex-facet incidences, and then uses the algorithm of Herrmann et al.~\cite{HerJosPfe-10} to construct the bounded subcomplex from the vertex-facet incidences. The algorithm of Herrmann et al. takes time $O(N^2\beta\phi + N^3\phi^2)$, where $N$ is the number of vertices of the given polyhedron (here at most $O(n^d)$), $\beta$ is the number of vertex-facet incidences (at most $O(n^{d+1})$, and $\phi$ is the number of faces of the bounded subcomplex (again, at most $O(n^d)$).
\item If the input is not in general position, the same method takes time $O(n^d L + n^{2d^2+3d})$.
\item If the input is in general position, we can alternatively construct the bounded subcomplex in time $O(n^{2d}L)$. The method is to maintain a list of faces of the bounded subcomplex, initially containing all the vertices. Then, for each face $f$ added to the list, and each vertex $v$, we compute the affine hull of $f\cup v$, use linear programming subproblems to determine whether this affine hull lies on the boundary of the polyhedron and has a vertex of maximum value, and if so check that the face contained in this affine hull is distinct from the ones already discovered. Whenever we find a new face in this way we add it to the list. There are $O(n^{2d})$ face-vertex pairs, each taking time $O(L)$ to check, so the time bound is as given.
\item If the input is not in general position, the same method takes time $O(n^{d^2+d}L)$.
\end{itemize}

When $d$ is not given as input (as seems more likely to occur in most applications of this problem), constructing the bounded subcomplex becomes more complicated. It is possible to interleave the construction of vertices (as the solutions of linear programming subproblems with decreasing dimension) with the construction of faces of increasing dimension, but how can we tell when we have found everything? An answer is provided by the following lemma, 
\new {which we illustrate in Figure~\ref {fig:nogap}.}

\eenplaatje[scale=1.5] {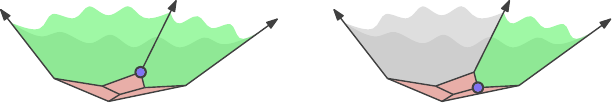} {Lemma~\ref {lem:nogap} states that if there exists a vertex of $P$ that is the minimum on the intersection of $d$ hyperplanes but no set of less than $d$ hyperplanes, then there also exists a vertex of $P$ that is the minimum on the intersection of $d-1$ hyperplanes but no set of less than $d-1$ hyperplanes. In the figure, $d=2$.}

\begin{lemma} \label {lem:nogap}
Let $H$ be a given set of halfspaces, let $P$ be \new {the} $D$-dimensional polyhedron formed as the intersection of the halfspaces in $H$, let $d\le D$ be a number, and let $\ell$ be a linear function.
For any given set $S$ of halfspaces in $H$, let $P_S$ denote the intersection of $P$ with the hyperplanes bounding \new {the} halfspaces in $S$.
Let $V$ be the set of all vertices of $P$ of the form $\min P_S$ for sets $S$ of at most $d$ halfspaces in $H$, and $F$ be a set of bounded faces of $P$ that includes the empty face. Suppose additionally that the following three conditions are all true:
\begin{enumerate}
\item For each set $S$ of at most $d+1$ halfspaces in $H$, $\min_\ell P_S$ belongs to $V$.
\item Every face in $F$ has dimension at most $d$.
\item For every pair $(f,v)$ of a face $f$ in $F$ and a vertex $v$ in $V$, with $A$ being the affine hull of $f\cup v$, one of the following three possibilities is true:
\begin{enumerate}
\item $A$ contains an interior point of $P$ and $d<D-1$,
\item Some face $f'\in F$ has affine hull $A$, or
\item $\ell$ is unbounded on $A\cap P$.
\end{enumerate}
\end{enumerate}
Then $F$ is the bounded subcomplex of the given polyhedron.
\end{lemma}

\begin{proof}
We assume for a contradiction that there is a bounded face $f$ that does not belong to $F$; among all such faces, let $f$ have the minimum possible dimension, and let $v=\max_\ell f$. If $f$ had dimension at most $d$, then by Lemma~\ref{lem:vertex-set-corr}, $v$ would belong to $V$, and the pair $(f',v)$ would violate condition (3) of the lemma where $f'$ is any facet of $f$ nonincident to $v$. We can assume without loss of generality that the lower link of $v$ (the subset of the link consisting of the faces whose values in $\ell$ are entirely less than that of the vertex itself) has dimension exactly $d$: it has dimension at least $d$, because it contains a face of dimension $\dim f - 1$ corresponding to $f$, and if it had a higher dimensional face then we could reduce the dimension by one unit per step by moving from $f$ to a facet of this higher dimensional face; condition (2) of the lemma ensures that, at each step, we continue to have a face $f$ that does not belong to $F$. But, if the lower link of $v$ has dimension exactly $d$, then by Lemma~\ref{lem:vertex-set-corr}, $v$ would be representable as $\min_\ell P_S$ for a set $S$ of at most $d+1$ halfspaces of $H$, violating condition (1) of the lemma. This contradiction shows that no such face $f$ exists.
\end{proof}

\begin{theorem}
Given a polyhedron whose bounded complex has dimension $d$ (with $d$ unknown to the algorithm) we can construct the bounded complex in time $O(n^{d^2+d}L)$. If the input is in general position, we can construct the bounded complex in time $O(n^{2d}L)$.
\end{theorem}

\begin{proof}
We apply the algorithm for known $d$, for increasing values of $d$, until the condition of the lemma is met. Testing the condition takes time $O(n^{d^2+d}L)$, or $O(n^{2d}L)$ for inputs in general position, matching the time for the known-$d$ algorithm on the correct value of $d$.
\end{proof}

It would be of interest to determine whether the method of Herrmann et al.~\cite{HerJosPfe-10}, which uses many fewer linear programs at the expense of greater running time in the other parts of the algorithm, can be adapted to the case where $d$ is unknown. In all cases the solution is found within an amount of time and a number of linear programming subproblems that is polynomial for any fixed value of~$d$. 

\section{Discussion}

We have shown that, for any fixed $d$, bounded subcomplexes of polyhedra that have dimension at most $d$ have polynomial complexity. Our bounds on the numbers of vertices of the bounded subcomplex are tight for small values of $d$. However, our bounds on the numbers of faces of higher dimensions do not appear to be tight. In other bounds on the complexity of polytopes, it is possible to considerably simplify the problem by assuming that the input is in general position; for our problem, too, such an assumption would be very helpful (it would lower the exponent of the polynomial from $d^2$ to $d$), but it is not always possible to perturb the input into general position without changing its dimension. Nevertheless, we would like to know whether instances that are not in general position have the same complexity as, or higher complexity than, instances that are in general position.

For the problem of constructing tight spans, our results are even more unsatisfactory, because the number of halfspaces is already quadratic in the number of points of the input metric space. For instance, when $d=2$, we get a bound of $O(n^8)$ on the complexity of the tight span of an $n$-point metric space. This stands in contrast with the bounds from another paper of $O(n)$ on the combinatorial complexity of the tight span and $O(n^2)$ time to construct it (optimal since the input distance matrix has size $O(n^2)$) under a stronger two-dimensionality assumption, that the tight span is homeomorphic to a subset of the plane~\cite{Epp-09}. Obtaining tighter bounds for $d$-dimensional tight spans, or even for the case $d=2$ without the requirement that the tight span form a planar set, would be of interest.

\subsection*{Acknowledgements}

This work was supported in part by NSF grant
0830403 and by the Office of Naval Research under grant
N00014-08-1-1015.

%\newpage
{ \raggedright
  \bibliographystyle{abuser}
  \bibliography{bounded-faces}
  \balance
}

\end {document}